\documentclass[a4paper]{article}

\usepackage{longtable}
\usepackage{amsmath,amsfonts,eucal,amsthm,amssymb,}
\usepackage{cite}

\usepackage{pgf}
\usepackage{tikz}
\usetikzlibrary{arrows,automata}
\usepackage{verbatim}
\usepackage{ifpdf}

\newtheorem{proposition}{Proposition}
\newtheorem{corollary}{Corollary}
\newtheorem{hypothesis}{Hypothesis}
\newtheorem{lemma}{Lemma}
\newtheorem{question}{Question}

\def\FP{\textit{Food Producer }}
\def\AP{\textit{Arms Producer }}

\def\MP{\textit{Medicine Producer }}

\def\Fo{\textit{Food }}

\def\calA{\mathcal{A}}
\def\calR{\mathcal{R}}

\def\calB{\mathcal{B}}

\def\R{\mathbb{R}}
\ifpdf
\usepackage[bookmarks=false,%
pdfstartview=FitH,linkbordercolor={0.5 1 1},%
citebordercolor={0.5 1 0.5},unicode,%
pagebackref,hyperindex]{hyperref}%
\else
\fi

\newcommand{\zer}{\hphantom{-}0}

\title{Sequences of Arbitrages%
}
\author{Victor Kozyakin\thanks{Institute for Information Transmission Problems, Moscow, Russia}
\and Brian O'Callaghan\thanks{University College Cork, Ireland}
\and Alexei Pokrovskii\thanks{University College Cork, Ireland}}
\begin{document}
\date{}
\maketitle

\begin{abstract}
The goal of this article is to understand some interesting features
of sequences of arbitrage operations, which look relevant to
various processes in Economics and Finances.

In the second part of the paper, analysis of sequences of
arbitrages is reformulated in the linear algebra terms. This admits
an elegant geometric interpretation of the problems under
consideration linked to the asynchronous systems theory. We feel
that this interpretation will be useful in understanding more
complicated, and more realistic, mathematical models in economics.



\end{abstract}

\tableofcontents

\section{Motivation}
Consider a mini-economy that involves only three producers. Each
producer produces one of three goods: either \textit{Food}, or
\textit{Arms}, or \textit{Medicine}. The economical activity is
reduced to the following three pair-wise barter operations:
\[
Food  \leftrightarrows Arms, \quad Food  \rightleftarrows Medicine,
\quad Arms \rightleftarrows Medicine.
\]
Suppose that the goods that are produced by each producer are
measured in some units, and the corresponding (strictly positive)
exchange rates, $r_{F, A}$, $r_{FM}$, $r_{AM}$, are well defined.
  That is, one
unit of \Fo can be exchanged for  $r_{F, A}$ units of \textit{Arms}.
The rates related to the inverted arrows are reciprocal:
\begin{equation}\label{rec3}
r_{AF}=\frac{1}{r_{FA}},\quad r_{MF}=\frac{1}{r_{FM}},
\quad r_{MA}=\frac{1}{r_{AM}}\ .
\end{equation}
We treat the triplet
\begin{equation}\label{prr3}
 \left(
 r_{FA}, \, r_{FM}, \,  r_{AM}
 \right)
 \end{equation}
  as the ensemble of {\em principal exchange rates.}

We suppose that, prior to a reference time moment $0$, each
producer knows only its own exchange rates: \FP does not know the
value of $r_{AM}$, \AP is unaware of $r_{FM}$, and \MP is unaware
of $r_{FA}$. We are interested in the case when the initial rates
are unbalanced in the following sense. By assumption, \FP can
exchange one unit of \Fo for $r_{FA}$ units of \textit{Arms}. Let
us suppose that unbeknownst to him the exchange rate between \MP
and \AP is such that the \FP could make a profit by first
exchanging one unit of \Fo for $r_{FM}$ units of \textit{Medicine}
and then exchanging these for \textit{Arms}. The inequality which
guarantees that \FP can take this advantage is that the product
$r_{FM}r_{MA}$ is greater than $r_{FA}$:
\begin{equation}\label{unb1}
r_{FM}\cdot r_{MA}>r_{FA}.
\end{equation}

Let us consider the situation when the inequality \eqref{unb1}
holds, and, after the reference time moment $0$, one of three
producers become aware about the third exchange rate. The evolution
of our economy depends on the detail {\em which producer is the
first to discover the information concerning the third exchange
rate}. The following three cases are relevant.

\paragraph{Case 1.}
\FP becomes aware of the value of the rate $r_{AM}$.
Therefore, \FP contacts \AP and makes a
request to increase the rate $r_{FA}$ to the new fairer value
\[
r^{new}_{FA}= r_{FM}\cdot r_{MA} = \frac{r_{FM}}{r_{AM}}.
\]
The reciprocal exchange rate  $r_{AF}$ is also to be
adjusted to the new level:
\[
r^{new}_{AF}=\frac{1}{r^{new}_{FA}}.
\]
The result is that the principal exchange rates become balanced
at the levels:
\[
r^{new}_{FA}=\frac{r_{FM}}{ r_{AM}}, \quad r_{FM}, \quad
r_{AM}.
\]

\paragraph{Case 2.}
\AP is the first to discover the third exchange rate $r_{FM}$. By
\eqref{rec3}, inequality \eqref{unb1} may be rewritten as
\[
\frac{r_{FM}} {r_{AM}}<\frac{1}{r_{AF}},
\]
which is, in turn,  equivalent to
\[
r_{AF}\cdot r_{FM}>r_{AM}.
\]
In this case \AP could do better by first exchanging \textit{Arms}
for \textit{Food},  and then  by exchanging this \textit{Food} for
\textit{Medicine}. Therefore, \AP requests adjustment of the rate
$r_{AM}$ to the value
\[
r^{new}_{AM}= r_{AF}\cdot r_{FM}= \frac{r_{FM}}{r_{FA}}.
\]
 In terms of the principal exchange rates  the outcome is
 that the economy is adjusted to the following balanced rates:
\[
r_{FA}, \quad r_{FM}, \quad r^{new}_{AM}= \frac{r_{FM}}{r_{FA}}.
\]

\paragraph{Case 3.}
\MP is the first to discover the third exchange rate $r_{FA}$. The
inequality \eqref{unb1} may be rewritten as
\[
r_{MA}\cdot r_{AF}>r_{MF}.
\]
 Thus, \MP requests adjustment of the rate  $r_{MF}$
to
\[
r^{new}_{MF}= r_{MA}\cdot r_{AF}.
\]
In this case the principal exchange rates become balanced at the
levels: \[\label{outc3} r_{FA}, \quad r^{new}_{FM}=  {r_{FA}\cdot
r_{AM}}, \quad r_{AM}.
\]

{\em After an adjustment of the principal exchange rate
\eqref{prr3}, following revealing an additional information as
described in any one of the cases  1--3 above, the exchange rates
become balanced, and this is the end of evolution of the
mini-economy with three producers. Our motivation to proceed with
this project was to understand possible scenarios of evolution of a
similar mini-economy with four producers.}

\section{Economical Aspects}
\subsection{\textit{FARM}-economy}

Consider the economy \textit{``FARM''} that includes four
producers, which produce \textit{Food}, \textit{Arms},
\textit{Rellics} and \textit{Medicine}. The economical activity is
described by six pair-wise barter operations:
\[
\begin{array}{ccc}
Food  \leftrightarrows Arms,& Food  \rightleftarrows Relics,& Food \rightleftarrows
Medicine,\\
Arms   \rightleftarrows Relics,& Arms \rightleftarrows  Medicine, &
Relics   \rightleftarrows Medicine.
\end{array}
\]
The goods that are produced by each producer are measured
in some units, and the  exchange rates
\[
\begin{array}{llllll}
r_{FA},& r_{FR},&r_{FM}, &r_{AF},& r_{AM},&r_{RM},\\
r_{RF},& r_{RA},&r_{RM},& r_{MF},& r_{MA},&r_{MR}
\end{array}
\]
are well defined.
The rates related to the inverted arrows are reciprocal:
\begin{equation}\label{rec}
\begin{array}{ccc}
\displaystyle r_{AF}=\frac{1}{r_{FA}},&
\displaystyle r_{RF}=\frac{1}{r_{FR}},&
\displaystyle r_{MF}=\frac{1}{r_{FM}}, \\[4mm]
\displaystyle r_{RA}=\frac{1}{r_{AF}},&
\displaystyle r_{MA}=\frac{1}{r_{AM}}, &
\displaystyle r_{MR}=\frac{1}{r_{RM}}.
\end{array}
\end{equation}
Our economy may be described by the ensemble of six {principal
exchange rates}
\begin{equation}\label{excr}
 \calR=\left(r_{FA},\ r_{FR},\ r_{FM},\
 r_{AR},\ r_{AM},\ r_{RM}
   \right)
\end{equation}
together with relationships \eqref{rec}.

The following characterization of
balanced exchange rates \eqref{excr} (that is, the exchange rates,
such that no one producer could do better when buying a certain good
through a mediator) is convenient.

\begin{proposition}\label{balp}
An ensemble
\[
\calR=\left(r_{FA},\ r_{FR} ,\ r_{FM} ,\ r_{AR} ,\ r_{AM} ,\
r_{RM}\right)
\]
of the principal exchange rates is balanced if and only if the relationships
\begin{eqnarray}
r_{FA}\cdot r_{AR} &= & r_{FR},\nonumber\\
r_{AR}\cdot r_{RM} &= & r_{AM},\label{invss}\\
r_{FA}\cdot r_{AR}\cdot  r_{RM} &= & r_{FM} \nonumber
\end{eqnarray}
hold.
\end{proposition}

\begin{proof}
This assertion can be proved by inspection.
\end{proof}

\subsection{Arbitrages}
 Let us suppose that initially each producer is aware only of three its
own exchange rates. For instance, \FP knows
only the rates
\begin{equation}\label{Frates}
r_{FA},\quad r_{FR},\quad r_{FM}.
\end{equation}

{\em We are interested in the case when the rates
\[
r_{FA}, \ r_{FR} ,\ r_{FM} ,\  r_{AR} ,\ r_{AM} ,\ r_{RM}
\]
are unbalanced}.

For instance, let us suppose that \FP can make profit by first,
exchanging one unit of \textit{Food} for $r_{FM}$ units of
\textit{Medicine}, and then by exchanging this \textit{Medicine}
for \textit{Arms}. Mathematically this means that the product
$r_{FM}\cdot r_{MA}$ is greater than $r_{FA}$:
\begin{equation}\label{unb}
r_{FM}\cdot r_{MA}>r_{FA}.
\end{equation}
Suppose further, that somebody makes \FP aware of the value
$r_{AM}$, and, therefore, about the inequality \eqref{unb}. \FP
makes a request that \AP should increase the exchange rate $r_{FA}$
to the new fairer value
\[
r^{new}_{FA}= {r_{FM}}\cdot r_{MA}= \frac{r_{FM}}{r_{AM}}.
\]
Along with the adjustment of the exchange rate $r_{FA}$, the
reciprocal rate  $r_{AF}$, should be adjusted to
\[
r^{new}_{AF}=\frac{1}{r^{new}_{FA}}.
\]
We call this procedure {\em $FAM$-arbitrage}, and we use the
notation $\calA_{FAM}$ to represent it. We denote by
$\calR\calA_{FAM}$ the ensemble of the new principal exchange
rates:
\[
\calR^{new}=\calR\calA_{FAM}=\left(r^{new}_{FA},\ r_{FR},\ r_{FM},\
r_{AR},\ r_{AM},\ r_{RM}
   \right).
\]
We also use the notation $\calR\calA_{FAM}$ in the case when the
inequality \eqref{unb} does not hold. In this case, of course,
$\calR\calA_{FAM}=\calR$, and we say that Arbitrage $\calA_{FAM}$
is {\em not active} in the later case.

This particular arbitrage is an example of the 24 possible
arbitrages listed in Table~\ref{tab1} in Subsection \ref{tabless}.

{\em The principal distinction of the \textit{FARM}-economy from the
economy with only three
 producers (as described in Motivation)  is that applying a single arbitrage
procedure would not necessarily result in bringing the economy to a
balance.}

\subsection{The Hypothesis}

One can apply arbitrages from Table~\ref{tab1} sequentially in any
order and to any initial exchange rates $\calR$.
A situation that we have in mind is the following. Suppose that
there exists \textit{Arbiter} who has access to the current ensemble $\calR$.
This \textit{Arbiter} could provide information to the producers in any order he
wants, thus activating the {\em chain} (or {\em superposition}) of corresponding
arbitrages.
The principal question is:
\begin{question}\label{que1}
How powerful is
\textit{Arbiter}?
\end{question}

 The short answer is: {\em \textit{Arbiter} is surprisingly
powerful; possibly, \textit{Arbiter} is almighty}.

Let us explain at a more formal level what we mean.

For a finite chain of arbitrages
$
{\bf A} = \calA_1 \ldots \calA_n ,
$
 and for a  given ensemble
$\calR$ of initial exchange rates, we denote by
\begin{equation}\label{arbsec}
\calR{\bf A}=\calR\calA_1 \ldots \calA_n
\end{equation}
 the resulting ensemble of principal exchange rates.
If $\calR$ is balanced, then $\calR{\calA}=\calR$ for any individual
arbitrage, and therefore $\calR{\bf A}=\calR$ for any chain
\eqref{arbsec}. If, on the contrary, $\calR$ is not balanced, then
different chains \eqref{arbsec} of arbitrages could result at
different balanced or unbalanced ensembles of principal exchange
rates. Denote by $S(\calR)$ the collection of the sets $\calR{\bf
A}$ related to all possible sequences \eqref{arbsec}. Denote also by
$S^{bal}(\calR)$ the subset of $S(\calR)$, that includes only
balanced exchange rates ensembles. Our principal observation is the following.

{\em For a typical unbalanced ensemble $\calR$ the set
$S^{bal}(\calR)$ is unexpectedly reach; therefore \textit{Arbiter},
who prescribes a particular sequence of arbitrages, is an
unexpectedly powerful figure.}

To avoid cumbersome notations and technical details when
providing a rigorous formulation
of this observation, we concentrate
on the simplest example of the initial ensemble.
Let us consider the ensemble
\begin{equation}\label{dist}
 \calR_{\alpha}=\left(\alpha \cdot \bar{r}_{FA},\
 \bar{r}_{FR},\ \bar{r}_{FM},\ \bar{r}_{AR},\ \bar{r}_{AM},\ \bar{r}_{RM}
   \right),
\end{equation}
  where $\alpha>1$ and $\bar\calR$
  is a given balanced ensemble of principal exchange rates.
   The ensemble  \eqref{dist} is not balanced.
   A possible origination of the ensemble \eqref{dist} may be commented as follows.
   Let us suppose that the underlying balanced rates
   \begin{equation}\label{distbar}
 \bar\calR=\left(\bar{r}_{FA},\ \bar{r}_{FR},\
  \bar{r}_{FM},\ \bar{r}_{AR},\ \bar{r}_{AM},\ \bar{r}_{RM}
   \right)
\end{equation}
had been in operation up to a certain reference time moment $0$. At this moment
$\tau$ the \FP has decided to increase his price for
\textit{Arms} by a factor $\alpha>1$. A natural specification of Question
1 is the following:
\begin{question}\label{que2}
To which balanced rates can \textit{Arbiter} now bring the
\textit{FARM}-economy?
\end{question}

The possible general structure of elements
from the corresponding sets $S(\calR_{\alpha})$  and
$S^{bal}(\calR_{\alpha})$ is easy to describe.
To this end we denote by $T_{\alpha}$ the collection of all six-tuples of the form
\begin{equation}\label{prod}
 \left(\alpha^{n_{1}} \cdot \bar{r}_{FA} ,\ \alpha^{n_{2}} \cdot \bar{r}_{FR},\ \alpha^{n_{3}} \cdot \bar{r}_{FM} ,
 \ \alpha^{n_{4}}\cdot \bar{r}_{AR},\ \alpha^{n_{5}}\cdot \bar{r}_{AM} ,\ \alpha^{n_{6}}\cdot \bar{r}_{RM}
 \right),
 \end{equation}
 where $n_{i}$ are integer numbers (positive, negative or zero).
  We also denote by $T^{bal}_{\alpha}$ the subset of
elements of $T_{\alpha}$, which satisfy the relationships
\begin{eqnarray}
n_{1}+ n_{4} &= & n_{2},\nonumber\\
n_{4}+ n_{6} &= & n_{5},\label{balda}\\
n_{1}+ n_{4} +n_{6}&= & n_{3}.\nonumber
\end{eqnarray}

\begin{proposition}\label{rep1P}
The inclusions
\begin{equation}\label{talp}
S(\calR_{\alpha})\subset T_{\alpha},
\end{equation}
and
\begin{equation}\label{palpb}
S^{bal}(\calR_{\alpha})\subset T^{bal}_{\alpha}
\end{equation}
hold.
\end{proposition}

\begin{proof}
The ensemble \eqref{distbar}  belongs to $T$. To verify
\eqref{talp} we show that the set $T_{\alpha}$ is invariant with
respect to each arbitrage $\calA$ from Table~\ref{tab1}. This
statement can be checked by inspection. Let us, for instance, apply
to a six-tuple \eqref{prod} the first arbitrage $\calA_{FAR}$.
Then, by definition, either this arbitrage is inactive, or it
changes the first component $\alpha^{n_{1}} \cdot \bar{r}_{FA} $ of
\eqref{prod} to the new value
\begin{equation}\label{bb}
r^{new}_{FA}= \frac{\alpha^{n_{2}} \cdot \bar{r}_{FR}}
{\alpha^{n_{4}}\cdot \bar{r}_{AR}}= \alpha^{n_{2}-n_{4}}\cdot
\frac{\bar{r}_{FR}}{ \bar{r}_{AR}}.
\end{equation}
However, the ensemble ${\bar\calR}$ is balanced, and,  by the first
equation \eqref{invss}, $\frac{\bar{r}_{FR}}{
\bar{r}_{AR}}=\bar{r}_{FA}$. Therefore, \eqref{bb} implies that
the ensemble $\bar\calR\calA_{FAR}$ also may be represented in the
form \eqref{prod}. We have proved the first part of the proposition,
related to the set $S(\calR_{\alpha})$.

The inclusion \eqref{palpb} follows now from Proposition \ref{balp}.
\end{proof}

Proposition \ref{rep1P} in no way answers Question \ref{que2}. This
proposition, however, allows us to reformulate this question in a
more ``constructive'' form:

\begin{question}\label{que3}
How big is the set $S^{bal}(\calR_{\alpha})$, comparing with the
collection  $T^{bal}_{\alpha}$  of all elements that satisfy
restrictions imposed by Proposition \ref{rep1P}?
\end{question}

The naive expectation would be that the set
$S^{bal}(\calR_{\alpha})$,
 is finite and, at least for the values of $\alpha$ close to 1, all elements of $S^{bal}(\calR_{\alpha})$ are
close to $\bar\calR$. However, some geometrical reasons,
along with results of extensive numerical experiments have convinced
us that the following statement, describing an unexpected feature of
the power of \textit{Arbiter}, is true.

\begin{hypothesis}\label{arbH}
The set
 $S^{bal}(\calR_{\alpha})$ coincides with $T^{bal}_{\alpha}$:
 \begin{equation}\label{palpe}
S^{bal}(\calR_{\alpha})= T^{bal}_{\alpha}.
\end{equation}
 Loosely speaking, this hypothesis means that \textit{Arbiter} is
 almighty.
 \end{hypothesis}


\subsection{Observations in Support of Hypothesis 1}\label{mrSS}

\begin{proposition}\label{arbP}
The set
 $S^{bal}(\calR_{\alpha})$ includes infinitely many different  ensembles.
 For instance, it contains the ensembles
  \begin{equation}\label{prodex}
 \left(\alpha \cdot \bar{r}_{FA} ,\ \alpha^{1-n} \cdot \bar{r}_{FR},\ \alpha \cdot \bar{r}_{FM} ,
 \
 \alpha^{-n}\cdot \bar{r}_{AR},\ \bar{r}_{AM} ,\ \alpha^{n}\cdot \bar{r}_{RM}
 \right),
 \end{equation}
 where $n$ is an arbitrary positive integer number.
\end{proposition}

\begin{proof}
It is sufficient to prove the ``for instance'' part. Consider
the chain
\[
{\bf A}= \calA_{RFM}\calA_{FMR}\calA_{MRA}
\calA_{FRM}\calA_{MFR} \calA_{RMA}.
\]
By $\hat{\bf A}^{n}$ we denote concatenation of $p$ copies
of $\hat{\bf A}$.
By inspection, for any $n=1,2,\ldots $, the ensemble \eqref{prodex}
can be generated by the chain of arbitrages $ {\bf
A}^{n}\calA_{FMA} .$
\end{proof}

To formulate some further observation in support of the Hypothesis
\ref{arbH}  the following corollary of the second part
of Proposition \ref{rep1P} is useful.

\begin{corollary}\label{rep1C}
The set $T^{bal}_{\alpha}$ coincides with the totality of all
six-tuples that may be written as
\begin{equation}\label{cdd}
\left(\alpha^{i}\cdot \bar{r}_{FA},\ \alpha^{i+j}\cdot
 \bar{r}_{FR},\ \alpha^{i+j+k}\cdot\bar{r}_{FM},\
 \alpha^{j}\cdot\bar{r}_{AR},\
\alpha^{j+k}\cdot\bar{r}_{AM},\ \alpha^{k}\cdot\bar{r}_{RM}
   \right),
\end{equation}
where $i,j,k$ are independent integer numbers.
\end{corollary}

Thus, using \eqref{cdd}, the ensembles from $T^{bal}_{\alpha}$ may
be uniquely coded by triplets $(i,j,k)$. We measure magnitudes of
such triplets by the characteristic
\[
 \label{maget}
 \|(i,j,k)\|=\max\{|i|,|j|,|k|\}.
\]
We denote by $S_{N}(\calR_{\alpha})$ the subset of
$S(\calR_{\alpha})$ which contains the ensembles that can be generated by chains of arbitrages
\eqref{arbsec} with $1 \le n\le N$. We also denote by
$S^{bal}_{N}(\calR_{\alpha})$ the corresponding subset of
$S_{N}(\calR_{\alpha}).$

Hypothesis \ref{arbH} would follow from the following stronger hypothesis:

\begin{hypothesis}\label{arb1H}
For any $\alpha>1$ the set $S^{bal}_{12\nu-1}(\calR_{\alpha})$
contains all balanced ensembles \eqref{cdd} whose codes have
magnitudes not greater then $\nu$, while
$S^{bal}_{12\nu-2}(\calR_{\alpha})$ contains balanced ensembles
with all aforementioned codes, except from the following two: $\pm
(\nu,\nu,\nu)$.
\end{hypothesis}

We have verified numerically the last hypothesis for $\nu=1,2,3 . $

In the context of numerical experiments the key question
is:
\begin{question}
\label{queN}
{How fast the numbers of elements in the sets
$S_{N}(\calR_{\alpha})$ and $S^{bal}_{N}(\calR_{\alpha})$ increase
in $N$?}
\end{question}
Proposition \ref{grop} below and its corollary provide an
encouraging answer.

 For an element $\calR$ of the form \eqref{prod} we define it's magnitude
 as
 \[
 \label{mage}
 \|\calR\|=\max\{|n_{1}|,|n_{2}|,|n_{3}|,|n_{4}|,|n_{5}|,|n_{6}|\}.
 \]

 \begin{proposition}
 \label{grop}
 The rate of increase of the magnitude $\|\calR_{\alpha}{\bf A}\|$ in $N$ is sub-linear:
 there  $\lambda>0$ such that
 $\|\calR_{\alpha}{\bf A}\|\le \lambda N,$
 where $N$ is the length of the sequence ${\bf A}$.
 \end{proposition}

 The proof of this assertion is provided in the next section.

Now we formulate only a corollary of Proposition \ref{grop},
which is directly relevant to computational hardship
of calculating the sets $S_{N}(\calR_{\alpha})$ and
 $ S_{N}^{bal}(\calR_{\alpha})$ for large $N$.
For a given set $S$ we denote by $\# S$ the number of elements in
this set.
\begin{corollary}\label{groc}
The estimates
\[
\# S_{N}(\calR_{\alpha})\le \mu {N^6},   \quad
 \# S_{N}^{bal}(\calR_{\alpha})\le \mu_{bal} {N^3},
\]
 where $\mu, \mu_{bal}$ are some positive constant,
 hold.
\end{corollary}

{\em On the basis of this corollary, we expect the analysis of the
set $S_{N}(\calR_{\alpha})$ is doable for $N$ of the order of 100.}

We note another unexpected feature or the \textit{Arbiter}'s
power.  One can expect that sufficiently long and sufficiently
``diverse'' sequences \eqref{arbsec} should result in achieving
balanced rates. The following proposition shows that this is
wrong.

\begin{proposition} \label{32}
There exist a chain  of 32 arbitrages, which contains all 24
arbitrages from Table~\ref{tab1} (and all arbitrages are active),
such that the corresponding chain \eqref{arbsec} is periodic after
a transient part. This chain is given by
\[
\begin{array}{llllllllllllllll}
5& 7& 17& 5&  14& 12& 15& 18&  11& 4& 18& 6&  10& 3& 8& 20\\
 19& 1& 23& 19& 14& 22& 9& 24& 21& 14& 24& 20& 16& 13& 2&  6;
\label{24}
\end{array}
\]
here, for brevity, we listed the numbers of arbitrages from
Table~\ref{tab1}, instead of the arbitrages themselves.
\end{proposition}

To conclude this subsection, we note that the set
$S(\calR_{\alpha})$ is, in contrast to \eqref{palpe}, much smaller
than the totality $T_{\alpha}$ of all ensembles of the form
\eqref{prod}. In particular, the following assertion holds.

\begin{proposition}\label{VictorP}
The set $S(\calR_{\alpha})$ does not contain the six-tuples
\[
\label{dist2}
 \calR_{\alpha^n}=\left(\alpha^{n} \bar{r}_{FA},\ \bar{r}_{FR},\ \bar{r}_{FM},\ \bar{r}_{A,
R},\ \bar{r}_{AM},\ \bar{r}_{RM} \right)
\]
 for $n\not= -1,0,1.$
\end{proposition}

\newpage
\subsection{Tables}\label{tabless}

\begin{center}
\begin{longtable}{lccc}
\caption{List of arbitrages}\label{tab1}\\
Number&Arbitrage&Activation condition&Actions\\
\hline \\
1& $\calA_{FAR}$& $\frac{r_{FR}}{r_{AR}}>r_{FA}$ &
$r^{new}_{FA}= \frac{r_{FR}}{r_{AR}}$
\\[2mm]

2& $\calA_{FAM}$& $\frac{r_{FM}}{r_{AM}}>r_{FA}$ &
$r^{new}_{FA}= \frac{r_{FM}}{r_{AM}}$\\[2mm]

3& $\calA_{FRA}$& $r_{FA}\cdot r_{AR}>r_{FR}$ &
$r^{new}_{FR}= r_{FA}\cdot r_{AR}$\\[2mm]

4& $\calA_{FRM}$& $\frac{r_{FM}}{r_{RM}}>r_{FR}$ &
$r^{new}_{FR}= \frac{r_{FM}}{r_{MR}}$ \\[2mm]

5& $\calA_{FMA}$& $r_{FA}\cdot r_{AM}>r_{FM}$ &
$r^{new}_{FM}= r_{FA}\cdot r_{AM}$ \\[2mm]

6& $\calA_{FMR}$& $r_{FR}\cdot r_{RM}>r_{FM}$ &
$r^{new}_{FM}= r_{FR}\cdot r_{RM}$ \\[2mm]

7& $\calA_{AFR}$& $\frac{r_{FR}}{r_{AR}}<r_{FA}$ &
$r^{new}_{FA}= \frac{r_{FR}}{r_{AR}}$ \\[2mm]

8& $\calA_{AFM}$& $\frac{r_{FM}}{r_{AM}}<r_{FA}$ &
$r^{new}_{FA}= \frac{r_{FM}}{r_{AM}}$\\[2mm]

9& $\calA_{ARF}$& $\frac{r_{FR}}{ r_{FA}}>r_{AR}$ &
$r^{new}_{AR}= \frac{r_{FR}}{r_{AF}}$ \\[2mm]

10 & $\calA_{ARM}$& $\frac{r_{AM}}{r_{MR}}>r_{AR}$ &
$r^{new}_{AR}= \frac{r_{AM}}{r_{MR}}$ \\[2mm]

11 & $\calA_{AMF}$& $\frac{r_{FM}}{r_{AF}}>r_{AM}$ &
$r^{new}_{AM}= \frac{r_{FM}}{r_{AF}}$ \\[2mm]

12 & $\calA_{AMR}$& $r_{AR}\cdot r_{RM}>r_{AM}$ &
$r^{new}_{AM}= r_{AR}\cdot r_{RM}$ \\[2mm]

13& $\calA_{RFA}$& $r_{FA}\cdot r_{AR}<r_{FR}$ &
$r^{new}_{FR}= r_{FA}\cdot r_{AR}$\\[2mm]

14& $\calA_{RFM}$& $\frac{r_{FM}}{r_{RM}}<r_{FR}$ &
$r^{new}_{FR}= \frac{r_{FM}}{r_{MR}}$ \\[2mm]

15& $\calA_{RAF}$& $\frac{r_{FR}}{ r_{AF}}<r_{AR}$ &
$r^{new}_{AR}= \frac{r_{FR}}{r_{AF}}$ \\[2mm]

16 & $\calA_{RAM}$& $\frac{r_{AM}}{r_{MR}}<r_{AR}$ &
$r^{new}_{AR}= \frac{r_{AM}}{r_{MR}}$ \\[2mm]

17& $\calA_{RMF}$& $\frac{r_{FM}}{r_{FR}}>r_{RM}$ &
$r^{new}_{RM}= \frac{r_{FM}}{r_{FR}}$
\\[2mm]

18& $\calA_{RMA}$& $\frac{r_{AM}}{r_{AR}}>r_{RM}$ &
$r^{new}_{RM}= \frac{r_{AM}}{r_{AR}}$\\[2mm]

19& $\calA_{MFA}$& $r_{FA}\cdot r_{AM}<r_{FM}$ &
$r^{new}_{FM}= r_{FA}\cdot r_{AM}$ \\[2mm]

20& $\calA_{MFR}$& $r_{FR}\cdot r_{RM}<r_{FM}$ &
$r^{new}_{F,
M}= r_{FR}\cdot r_{RM}$ \\[2mm]

21 & $\calA_{MAF}$& $\frac{r_{FM}}{r_{AF}}<r_{AM}$ &
$r^{new}_{AM}= \frac{r_{FM}}{r_{AF}}$ \\[2mm]

22 & $\calA_{MAR}$& $r_{AR}\cdot r_{RM}<r_{AM}$ &
$r^{new}_{AM}= r_{AR}\cdot r_{RM}$ \\[2mm]

23& $\calA_{MRF}$& $\frac{r_{FM}}{r_{FR}}<r_{RM}$ &
$r^{new}_{RM}= \frac{r_{FM}}{r_{FR}}$
\\[2mm]

24& $\calA_{MRA}$& $\frac{r_{AM}}{r_{AR}}<r_{RM}$ &
$r^{new}_{RM}= \frac{r_{AM}}{r_{AR}}$\\[2mm]
\end{longtable}
\end{center}

\newpage
\begin{center}
\begin{longtable}{c p{1.6cm} p{1.7cm} l}
\caption{Optimal chains of (strong) arbitrages to reach 27 balanced
ensembles with the codes $(i,j,k)$ satisfying $|i|,|j|,|k|\le
1$}\label{tab2}\\
 Number& Strategy's length &Balanced outcome's code& Optimal sequence of strong arbitrages\\
\hline \\
1& 1& $(0,0,0)$&$ 2$\\[2mm]

2& 2&$(1,-1,0)$ & $7, 10$\\
3&  2& $(1,1,0)$ & $3, 6$\\[2mm]

4& 3 &$(1, -1, 1)$ & $5, 7, 12$\\
5&  3& $(1, 0, -1)$ & $3, 9, 12$\\[2mm]

6&4&$(0, -1, 1)$&$7, 2, 3, 12$\\
7&4&$(0, 0, -1)$&$9, 1, 5, 12$\\
8 &4&$(0, 0, 1)$&$5, 1, 9, 12$\\
9&4&$( 0, 1, -1)$&$3, 2, 7, 12$\\[2mm]

10&5&$(0, -1, 0)$&$ 7, 2, 3, 6, 10$\\
11&5&$( 0, 1, 0)$&$3, 2, 6, 7, 10$\\
12 &5&$(1, -1, -1)$&$9, 12, 6, 7, 10$\\
13&5&$( 1, 0, 1)$&$5, 11, 3, 6, 10$\\
14 &5&$(-1, 1, -1)$&$3, 11, 5, 4, 8$\\[2mm]

15& 6 &$(-1, 0, 0)$ & $9, 1, 5, 4, 1, 10$\\
16&  6& $(-1, 1, 0)$ & $3, 2, 7, 4, 1, 10$\\[2mm]

17&7&$( -1, -1, 1)$&$ 7, 2, 3, 6, 2, 3, 12$\\
18&7&$( -1, 0, -1)$&$9, 1, 5, 4, 1, 5, 12$\\
19 &7&$(-1, 0, 1)$&$7, 2, 3, 8, 1, 9, 12$\\
20&7&$( -1, 1, -1)$&$9, 1, 5, 10, 2, 7, 12$\\
21 &7&$(-1, 1, 1)$&$5, 1, 9, 8, 1, 9, 12$\\[2mm]

22&8& $(-1, -1, 0)$& $7, 2, 3, 6, 2, 3, 6, 10$\\
 23&8&$(0, -1, -1)$&$9, 1, 5, 4, 12, 6, 7, 10$\\
24&8&$(0, 1, 1)$& $3, 2, 6, 9, 12, 6, 7, 10$\\
25&8& $(1, 1, 0)$&$3, 11, 5, 4, 12, 6, 7, 10$\\[2mm]

26 &11 & $(-1, -1, -1)$& 9, 1, 5, 4, 1, 5, 4, 12, 6, 7, 10\\
27 &11 & $(1, 1, 1)$ & 3, 11, 5, 4, 12, 6, 9, 12, 6, 7, 10\\

\end{longtable}
\end{center}

\section{Mathematical Background}
\subsection{Reformulation in the Linear Algebra Terms}
Analysis of sequences of arbitrages in \textit{FARM}-economy admits
an elegant geometric interpretation, to be discussed in this
section. Actually, we have used heavily this interpretation when
inventing and  proving results from Subsection \ref{mrSS} (although
many proves can be eventually rewritten without explicit references
to the geometrical interpretation). We also feel that this
interpretation will be useful in understanding  more complicated,
and more realistic, mathematical models in economics.

We use, as an auxiliary tool, a somehow stronger arbitrage
procedure. Let us begin with an example. Consider the combination
$(FAM)$. For a given $\calR$ we define the Strong
Arbitrage $ \hat\calA_{(FAM)}\calR $ as $\calA_{FAM}$ if the
inequality \eqref{unb} holds, and as
 $\calA_{AFM}$, otherwise.
 Note that in both cases the result in terms of principal rates is the same:
 the rate $r_{FA}$ is changed to
 $
 r_{FA}^{new}=\frac{r_{FM}}{r_{AM}}.
 $

The strong arbitrage $\hat\calA_{(FAM)}$ is the second entry in
Table~\ref{starbT} of the possible 12 strong arbitrages. The
meaning of a strong arbitrage is simple. This is just balancing a
corresponding ``sub-economy'' $(FAM)$ by changing the exchange rate
for a pair $F\leftrightarrows A.$

\begin{center}
\begin{longtable}{clc}
\caption{Strong arbitrages}\label{starbT}\\
 Number& Strong arbitrage& Action\\
\hline \\

1& $\hat\calA_{FAR}$& $\displaystyle r^{new}_{FA}=
\frac{r_{FR}}{r_{AR}}$
\\[2mm]
2& $\hat\calA_{FAM}$  & $\displaystyle r^{new}_{FR}=
\frac{r_{FM}}{r_{AM}} $
\\[2mm]
3& $\hat\calA_{FRA}$  & $\displaystyle r^{new}_{FR}= r_{F,
A}\cdot r_{AR}$
\\[2mm]
4& $\hat\calA_{FRM}$& $\displaystyle r^{new}_{FR}=
\frac{r_{FM}}{r_{RM}}$
\\[2mm]
5& $\hat\calA_{FMA}$ & $\displaystyle r^{new}_{FM}= r_{F,
A}\cdot r_{AM}$
\\[2mm]
6 &$\hat\calA_{FMR}$&
 $\displaystyle r^{new}_{FM}= r_{FR}\cdot r_{RM}$
\\[2mm]
7& $\hat\calA_{ARF}$&
  $\displaystyle r^{new}_{AR}= \frac{r_{FR}}{r_{FA}}$
\\[2mm]
8& $\hat\calA_{ARM}$&
 $\displaystyle r^{new}_{AR}= \frac{r_{AM}}{r_{RM}}$
\\[2mm]
9& $\hat\calA_{AMF}$&
 $\displaystyle r^{new}_{AM}= \frac{r_{FM}}{r_{FA}}$
\\[2mm]
10& $\hat\calA_{AMR}$&
 $\displaystyle r^{new}_{AM}= r_{AR}\cdot r_{RM}$
\\[2mm]
11& $\hat\calA_{RMF}$&
 $\displaystyle r^{new}_{RM}= \frac{r_{FM}}{r_{FR}}$
\\[2mm]
12& $\hat\calA_{RMA}$& $\displaystyle r^{new}_{RM}=
\frac{r_{AM}}{r_{AR}}$
\end{longtable}
\end{center}

\begin{proposition}
For any sequence of arbitrages \eqref{arbsec} and any initial
exchange rates $\calR$ there exist a chain $\hat{\bf A}=\hat\calA_1
\ldots \hat\calA_n $ of strong arbitrages, such that $\calR\hat{\bf
A}=\calR{\bf A}$. Conversely, for any chain $\hat{\bf A}=\hat\calA_1
\ldots \hat\calA_n $ of strong arbitrages and any initial exchange
rates $\calR$ there exist a sequence of arbitrages, such that
$\calR\hat{\bf A}=\calR{\bf A}$.
\end{proposition}

This proposition reduces investigation of the questions from the
previous subsection to investigation of analogous questions related
to sequences of strong arbitrages.

We define a correspondence to ensemble
\[
\calR= \left(r_{FA},\ r_{FR},\ r_{FM},\ r_{AR},\ r_{AM},\ r_{RM}
   \right)
\]
of principal exchange rates, and a column vector
$v=v(\calR)\in\R^{6}$ via the following procedure
\begin{eqnarray*}
v(\calR)&=& \left( v^{(1)}, v^{(2)}, v^{(3)},v^{(4)},v^{(5)},
v^{(6)}\right)\\
&= &  \left( \log r_{FA}, \log r_{AR}, \log r_{RM}, \log r_{FR},
\log r_{AM}, \log r_{FM} \right) .
\end{eqnarray*}

Now we relate a strong arbitrage, which has number $n$ in
Table~\ref{starbT}, a $6\times 6$ matrix $B_{n}$, $n=1,\ldots, 12$,
as follows:

{\scriptsize
\[
  B_{1} =
\left(
\begin{array}{rrrrrr}
\zer& \zer& \zer& \zer& \zer& 0\\ -1& 1& \zer& \zer& \zer& 0 \\ \zer& \zer& 1& \zer& \zer& 0\\
 1& \zer& \zer& 1& \zer& 0\\ \zer& \zer& \zer& \zer& 1& 0\\ \zer& \zer& \zer& \zer& \zer& 1
\end{array}
\right) , \qquad
 B_{2}=
\left(
\begin{array}{rrrrrr}
\zer& \zer& \zer& \zer& \zer& 0\\ \zer& 1& \zer& \zer& \zer& 0\\ \zer& \zer& 1& \zer& \zer& 0 \\
 \zer& \zer& \zer& 1& \zer& 0\\ -1& \zer& \zer& \zer& 1& 0\\ 1& \zer& \zer& \zer& \zer& 1
\end{array}
\right) ,
\]
\[
B_{3}= \left(
\begin{array}{rrrrrr}
 1& \zer& \zer& 1& \zer& 0\\ \zer& 1& \zer& 1& \zer& 0\\ \zer& \zer& 1& \zer& \zer& 0\\
  \zer& \zer& \zer& \zer& \zer& 0\\ \zer& \zer& \zer& \zer& 1& 0\\ \zer& \zer& \zer& \zer& \zer& 1
\end{array}
\right) , \qquad
 B_{4}= \left(
\begin{array}{rrrrrr}
 1& \zer& \zer& \zer& \zer& 0\\ \zer& 1& \zer& \zer& \zer& 0\\ \zer& \zer& 1& -1& \zer& 0\\
  \zer& \zer& \zer& \zer& \zer& 0\\ \zer& \zer& \zer& \zer& 1& 0\\ \zer& \zer& \zer& 1& \zer& 1
\end{array}
\right),
\]
\[
B_{5}= \left (\begin{array}{rrrrrr}
 1& \zer& \zer& \zer& \zer& 1\\ \zer& 1& \zer& \zer& \zer& 0\\ \zer& \zer& 1& \zer& \zer& 0\\
  \zer& \zer& \zer& 1& \zer& 0\\ \zer& \zer& \zer& \zer& 1& 1\\ \zer& \zer& \zer& \zer& \zer& 0
\end{array}
\right), \qquad
 B_{6}= \left
 (\begin{array}{rrrrrr}
 1& \zer& \zer& \zer& \zer& 0\\ \zer& 1& \zer& \zer& \zer& 0\\ \zer& \zer& 1& \zer& \zer& 1\\
  \zer& \zer& \zer&  1& \zer& 1\\ \zer& \zer& \zer& \zer& 1& 0\\ \zer& \zer& \zer& \zer& \zer& 0
\end{array}
\right),
\]
\[
B_{7}= \left(\begin{array}{rrrrrr}
  1& -1& \zer& \zer& \zer& 0\\ \zer& \zer& \zer& \zer& \zer& 0\\ \zer& \zer& 1& \zer& \zer& 0\\
   \zer& 1& \zer& 1& \zer& 0\\ \zer& \zer& \zer& \zer& 1& 0\\ \zer& \zer& \zer& \zer& \zer& 1
\end{array}\right),
\qquad
 B_{8}= \left(\begin{array}{rrrrrr}
  1& \zer& \zer& \zer& \zer& 0\\ \zer& \zer& \zer& \zer& \zer& 0\\ \zer& -1& 1& \zer& \zer& 0\\
   \zer& \zer& \zer& 1& \zer& 0\\ \zer& 1& \zer& \zer& 1& 0\\ \zer& \zer& \zer& \zer& \zer& 1
\end{array}\right),
\]
\[
B_{9}= \left(\begin{array}{rrrrrr}
  1& \zer& \zer& \zer& -1& 0\\ \zer& 1& \zer& \zer& \zer& 0\\ \zer& \zer& 1& \zer& \zer& 0\\
   \zer& \zer& \zer& 1& \zer& 0\\ \zer& \zer& \zer& \zer& \zer& 0\\ \zer& \zer& \zer& \zer& 1& 1
\end{array}\right),
\qquad
 B_{10}=\left(\begin{array}{rrrrrr}
   1& \zer& \zer& \zer& \zer& 0\\ \zer& 1& \zer& \zer& 1& 0\\ \zer& \zer& 1& \zer& 1& 0\\
    \zer& \zer& \zer& 1& \zer& 0\\ \zer& \zer& \zer& \zer& \zer& 0\\ \zer& \zer& \zer& \zer& \zer& 1
\end{array}\right),
\]
\[
B_{11}=\left(\begin{array}{rrrrrr}
  1& \zer& \zer& \zer& \zer& 0\\ \zer& 1& \zer& \zer& \zer& 0\\ \zer& \zer& \zer& \zer& \zer& 0\\
   \zer& \zer& -1& 1& \zer& 0\\ \zer& \zer& \zer& \zer& 1& 0\\ \zer& \zer& 1& \zer& \zer& 1
\end{array}\right),
\qquad
 B_{12}=\left(\begin{array}{rrrrrr}
  1& \zer& \zer& \zer& \zer& 0\\ \zer& 1& -1& \zer& \zer& 0\\ \zer& \zer& \zer& \zer& \zer& 0\\
   \zer& \zer& \zer& 1& \zer& 0\\ \zer& \zer& 1& \zer& 1& 0\\ \zer& \zer& \zer& \zer& \zer& 1
\end{array}\right).
\]}

\noindent We denote by $\calB$ the ensemble of these matrices.

\begin{proposition}
For any strong arbitrage $\hat\calA$ with number $k$ in
Table~\ref{starbT}, and any ensemble $\calR$ the equality $
v(\calR\hat\calA )=v(\calR)B_{k} $ holds.
\end{proposition}

\begin{corollary}
For any chain $ \hat{\bf A}=\hat{\calA}_{1}\ldots \hat{\calA}_{n} $
of strong arbitrages the relationship \[ \label{matr}
v(\hat{\calR\bf A})=v(\calR)\prod_{i=1}^{n}B_{k(i)}
\]
holds, where $k(i)$ is the number of Arbitrage $\calA_{i}$, $i=1,
\ldots, n$. In particular for an initial state of the form
\eqref{dist}, the vector $v(\calR_{\alpha})\hat{\bf A}$ may be
written as
\[
v(\bar\calR_{\alpha})+(\log \alpha)v \prod_{i=1}^{n}B_{k(i)},
\]
where
$
v= \left( 1,0,0,0,0,0 \right).
$
\end{corollary}

This proposition reduces analysis of sequences of strong
arbitrages to analysis of products of matrices $B$.

\subsection{A Special Coordinate System}
Proposition \ref{balp} implies
\begin{corollary}
The matrices $B_{i}$, $i=1,\ldots, 12, $ have a common invariant
subspace defined by
\begin{eqnarray}
v^{(1)}+ v^{(2)} &= & v^{(4)},\nonumber\\
v^{(2)}+ v^{(3)} &= & v^{(5)}, \nonumber \label{balda2}\\
v^{(1)}+ v^{(2)} +v^{(3)}&= & v^{(6)}.\nonumber
\end{eqnarray}
\end{corollary}

By this corollary there exists a substitution of variables  $Q$, such
that each matrix $Q B_{n}Q^{-1}$ has the block-triangular form:
\[
\label{E-QBQ}
    D_{n}:=Q^{-1}B_{n}Q=\left(\begin{array}{cc}
I& 0\\
F_{n}& G_{n}
\end{array}\right), \quad n=1,\ldots, 12.
\]

Here the matrices  $Q$ and $Q^{-1}$ may be chosen as follows:
{\scriptsize
\[
Q=\left(\begin{array}{rrr|rrr}
1& \zer& \zer& 1& \zer& 1\\
\zer& 1& \zer& 1& 1& 1\\
\zer& \zer& 1& \zer& 1& 1\\ \hline
\zer& \zer& \zer& 1& \zer& \zer\\
\zer& \zer& \zer&    \zer& 1& \zer\\
\zer& \zer& \zer&    \zer& \zer& 1
\end{array}\right),\
Q^{-1}=\left(\begin{array}{rrr|rrr}
1& \zer& \zer& -1& \zer& -1\\
\zer& 1& \zer& -1& -1  & -1\\
\zer& \zer& 1& \zer& -1& -1\\ \hline
\zer& \zer& \zer& 1& \zer& \zer\\
\zer& \zer& \zer& \zer& 1& \zer\\
\zer& \zer& \zer& \zer& \zer& 1
\end{array}\right).
\]
}
In the new coordinates the matrices $D_{n}:=Q^{-1}B_{n}Q$
take the form:
{\scriptsize
\[
D_{1}=\left(\begin{array}{rrr|rrr}
1& \zer& \zer& \zer& \zer& \zer\\
\zer& 1& \zer& \zer& \zer& \zer\\
\zer& \zer& 1& \zer& \zer& \zer\\ \hline
1& \zer& \zer& \zer& \zer& -1\\
\zer& \zer& \zer& \zer& 1& \zer\\
\zer& \zer& \zer& \zer& \zer& 1
\end{array}\right),\
D_{2}=\left(\begin{array}{rrr|rrr}
1& \zer& \zer& \zer& \zer& \zer\\
\zer& 1& \zer& \zer& \zer& \zer\\
\zer& \zer& 1& \zer& \zer& \zer\\ \hline
\zer& \zer& \zer&1& \zer& \zer\\
-1& \zer& \zer& 1& 1& 1\\
1& \zer& \zer& -1& \zer& \zer
\end{array}\right),
\]
\[
D_{3}=\left(\begin{array}{rrr|rrr}
1& \zer& \zer& \zer& \zer& \zer\\
\zer& 1& \zer& \zer& \zer& \zer\\
\zer& \zer& 1& \zer& \zer& \zer\\ \hline
\zer& \zer& \zer&\zer& \zer& \zer\\
\zer& \zer& \zer& \zer& 1& \zer\\
\zer& \zer& \zer& \zer& \zer& 1
\end{array}\right),\
D_{4}=\left(\begin{array}{rrr|rrr}
1& \zer& \zer& \zer& \zer& \zer\\
\zer& 1& \zer& \zer& \zer& \zer\\
\zer& \zer& 1& \zer& \zer& \zer\\ \hline
\zer& \zer& \zer&\zer& \zer& \zer\\
\zer& \zer& \zer& \zer& 1& \zer\\
\zer& \zer& \zer& 1& \zer& 1
\end{array}\right),
\]
\[
D_{5}=\left(\begin{array}{rrr|rrr}
1& \zer& \zer& \zer& \zer& \zer\\
\zer& 1& \zer& \zer& \zer& \zer\\
\zer& \zer& 1& \zer& \zer& \zer\\ \hline
\zer& \zer& \zer& 1& \zer& \zer\\
\zer& \zer& \zer& \zer& 1& 1\\
\zer& \zer& \zer& \zer& \zer& \zer
\end{array}\right),\
D_{6}=\left(\begin{array}{rrr|rrr}
1& \zer& \zer& \zer& \zer& \zer\\
\zer& 1& \zer& \zer& \zer& \zer\\
\zer& \zer& 1& \zer& \zer& \zer\\ \hline
\zer& \zer& \zer&1& \zer& 1\\
\zer& \zer& \zer& \zer& 1& \zer\\
\zer& \zer& \zer& \zer& \zer& \zer
\end{array}\right),
\]
\[
D_{7}=\left(\begin{array}{rrr|rrr}
1& \zer& \zer& \zer& \zer& \zer\\
\zer& 1& \zer& \zer& \zer& \zer\\
\zer& \zer& 1& \zer& \zer& \zer\\ \hline
\zer& 1& \zer&\zer& -1& -1\\
\zer& \zer& \zer& \zer& 1& \zer\\
\zer& \zer& \zer& \zer& \zer& 1
\end{array}\right),\
D_{8}=\left(\begin{array}{rrr|rrr}
1& \zer& \zer& \zer& \zer& \zer\\
\zer& 1& \zer& \zer& \zer& \zer\\
\zer& \zer& 1& \zer& \zer& \zer\\ \hline
\zer& \zer& \zer& 1& \zer& \zer\\
\zer& 1& \zer& -1& \zer& -1\\
\zer& \zer& \zer& \zer& \zer& 1
\end{array}\right),
\]
\[
D_{9}=\left(\begin{array}{rrr|rrr}
1& \zer& \zer& \zer& \zer& \zer\\
\zer& 1& \zer& \zer& \zer& \zer\\
\zer& \zer& 1& \zer& \zer& \zer\\ \hline
\zer& \zer& \zer& 1& \zer& \zer\\
\zer& \zer& \zer& \zer& \zer& \zer0\\
\zer& \zer& \zer& \zer& 1& 1
\end{array}\right),\
D_{10}=\left(\begin{array}{rrr|rrr}
1& \zer& \zer& \zer& \zer& \zer\\
\zer& 1& \zer& \zer& \zer& \zer\\
\zer& \zer& 1& \zer& \zer& \zer\\ \hline
\zer& \zer& \zer& 1& \zer& \zer\\
\zer& \zer& \zer& \zer& \zer& \zer\\
\zer& \zer& \zer& \zer& \zer& 1
\end{array}\right),
\]
\[
D_{11}=\left(\begin{array}{rrr|rrr}
1& \zer& \zer& \zer& \zer& \zer\\
\zer& 1& \zer& \zer& \zer& \zer\\
\zer& \zer& 1& \zer& \zer& \zer\\ \hline
\zer& \zer& -1& 1& 1& 1\\
\zer& \zer& \zer& \zer& 1& \zer\\
\zer& \zer& 1& \zer& -1& \zer
\end{array}\right),\
D_{12}=\left(\begin{array}{rrr|rrr}
1& \zer& \zer& \zer& \zer& \zer\\
\zer& 1& \zer& \zer& \zer& \zer\\
\zer& \zer& 1& \zer& \zer& \zer\\ \hline
\zer& \zer& \zer& 1& \zer& \zer\\
\zer& \zer& 1& \zer& \zer& -1\\
\zer& \zer& \zer& \zer& \zer& 1
\end{array}\right).
\]
}

\subsection{Key Graph of \textit{FARM}-economy}
The   South-East blocks  $G_{n}$, are
the following:
{\scriptsize
\[
G_{1}=\left(\begin{array}{rrr}
\zer& \zer& -1\\
\zer& 1& \zer\\
\zer& \zer& 1
\end{array}\right),
\quad
G_{2}=\left(\begin{array}{rrr}
1& 0&0\\
1& 1& 1\\
-1& \zer& \zer
\end{array}\right),
\quad
G_{3}=\left(\begin{array}{rrr}
\zer& \zer& \zer\\
\zer& 1& \zer\\
\zer& \zer& 1
\end{array}\right),
\]
\[
G_{4}=\left(\begin{array}{rrr}
\zer& \zer& \zer\\
\zer& 1& \zer\\
1& \zer& 1
\end{array}\right),
\quad
G_{5}=\left(\begin{array}{rrr}
1& \zer& \zer\\
\zer& 1& 1\\
\zer& \zer& \zer
\end{array}\right),\quad
G_{6}=\left(\begin{array}{rrr}
1& \zer& 1\\
\zer& 1& \zer\\
\zer& \zer& \zer
\end{array}\right),
\]
\[
G_{7}=\left(\begin{array}{rrr}
\zer& -1& -1\\
\zer& 1& \zer\\
\zer& \zer& 1
\end{array}\right),\quad
G_{8}=\left(\begin{array}{rrr}
1&\zer& \zer\\
-1& \zer&-1\\
\zer&\zer& 1
\end{array}\right),
\quad
G_{9}=\left(\begin{array}{rrr}
1& \zer& \zer\\
\zer& \zer& \zer\\
\zer& 1& 1\\
\end{array}\right),
\]
\[
G_{10}=\left(\begin{array}{rrr}
1& \zer& \zer\\
\zer& \zer& \zer\\
\zer& \zer& 1
\end{array}\right),\quad
G_{11}=\left(\begin{array}{rrr}
\hphantom{-}1& 1& 1\\
\zer& 1&\zer\\
\zer& -1& \zer
\end{array}\right),\quad
G_{12}=\left(\begin{array}{rrr}
1& \zer& \zer\\
\zer& \zer& -1\\
\zer&\zer& 1
\end{array}\right).
\]
}

Denote by $S$ the set of vectors which are mapped to zero by at
least one of the matrices $G_{n}$. By inspection these are the
vectors proportional to one of the following six vectors:
\begin{eqnarray*}
s_{1}=\left(1,0,1\right),&  s_{2}=\left(1,0,0\right),&
s_{3}=\left(0,0,1\right),\\
s_{4}=\left(1,1,1\right), &s_{5}=\left(0, 1, 0\right),&
s_{6}=\left(0, 1, 1\right).
\end{eqnarray*}

By definition the set $ \left\{\pm s_{1},\pm s_{2},\pm s_{3},\pm
s_{4},\pm s_{5},\pm s_{6}\right\} $ is transformed by any matrix
$G_{n}$ into itself. The  graph of the corresponding transitions
(see Fig.~\ref{F-1}) is essential for understanding our problem,
and we call it the {\em key graph of
 \textit{FARM}-economy}.

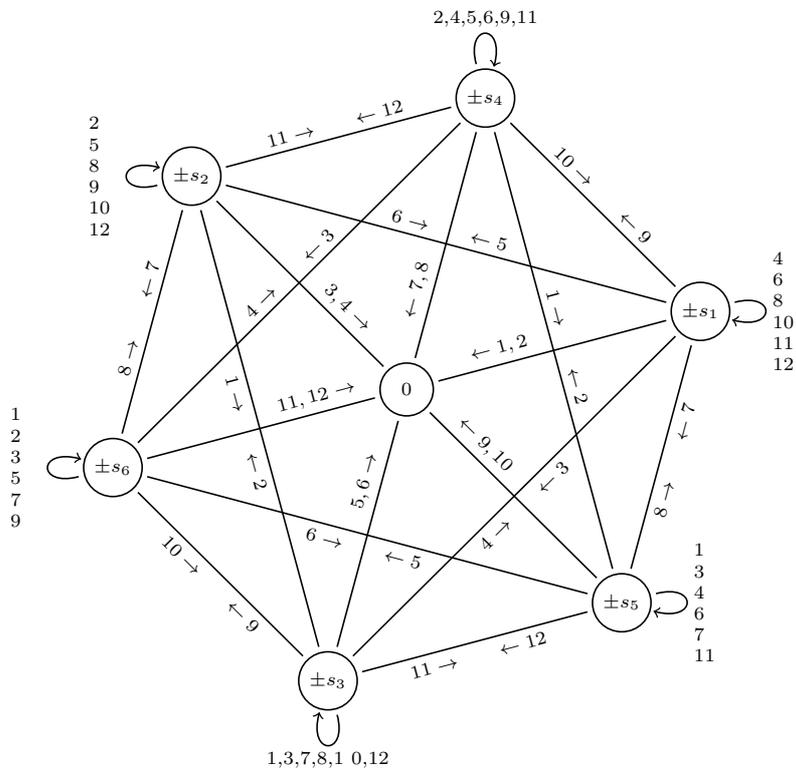
\begin{figure}[!htbp]
\centering\scriptsize
\begin{tikzpicture}[-,shorten <=2pt,shorten >=2pt,auto,node distance=5cm,
                    semithick]
 \tikzstyle{every state}=[fill=none,draw=black,text=black]
 \node[state] (0) at (0,0)  {$0$};
 \node[state] (1) at (15:4)  {$\pm s_{1}$};
 \node[state] (2) at (135:4) {$\pm s_{2}$};
 \node[state] (3) at (255:4) {$\pm s_{3}$};
 \node[state] (4) at (75:4)  {$\pm s_{4}$};
 \node[state] (5) at (315:4) {$\pm s_{5}$};
 \node[state] (6) at (195:4) {$\pm s_{6}$};

 \path (1) edge              node [sloped, text centered, above] {\llap{$\leftarrow1,2$~~~}} (0);
 \path (1) edge              node [sloped, text centered, below] {$4\rightarrow$\qquad$\leftarrow3$} (3);
 \path (1) edge [loop right] node {\parbox[c]{4mm}{4\\6\\8\\10\\11\\12}} (1);
 \path (1) edge              node [sloped, text centered, above] {$6\rightarrow$\qquad$\leftarrow5$} (2);
 \path (1) edge              node [sloped, text centered, below] {$8\rightarrow$\qquad$\leftarrow7$} (5);
 \path (1) edge              node [sloped, text centered, above] {$10\rightarrow$\qquad$\leftarrow9$} (4);

 \path (2) edge              node [sloped, text centered, below] {$1\rightarrow$\qquad$\leftarrow2$} (3);
 \path (2) edge [loop left]  node {\parbox[c]{4mm}{2\\5\\8\\9\\10\\12}} (2);
 \path (2) edge              node [sloped, text centered, above] {\rlap{~~~$3,4\rightarrow$}} (0);
 \path (2) edge              node [sloped, text centered, above] {$8\rightarrow$\qquad$\leftarrow7$} (6);
 \path (2) edge              node [sloped, text centered, above] {$11\rightarrow$\qquad$\leftarrow12$} (4);

 \path (3) edge [loop below] node {1,3,7,8,1\zer,12} (3);
 \path (3) edge              node [sloped, text centered, above] {\rlap{~~~$5,6\rightarrow$}} (0);
 \path (3) edge              node [sloped, text centered, below] {$10\rightarrow$\qquad$\leftarrow9$}  (6);
 \path (3) edge              node [sloped, text centered, below] {$11\rightarrow$\qquad$\leftarrow12$} (5);

 \path (4) edge              node [sloped, text centered, above] {$1\rightarrow$\qquad$\leftarrow2$} (5);
 \path (4) edge [loop above] node {\scriptsize{2,4,5,6,9,11}} (4);
 \path (4) edge              node [sloped, text centered, above] {\llap{$\leftarrow7,8$~~~}} (0);
 \path (4) edge              node [sloped, text centered, above] {$4\rightarrow$\qquad$\leftarrow3$} (6);

 \path (5) edge [loop right] node {\scriptsize{\parbox[c]{4mm}{1\\3\\4\\6\\7\\11}}} (5);
 \path (5) edge              node [sloped, text centered, below] {$6\rightarrow$\qquad$\leftarrow5$} (6);
 \path (5) edge              node [sloped, text centered, above] {\llap{$\leftarrow9,10$~~~}} (0);

 \path (6) edge [loop left]  node {\scriptsize{\parbox[c]{4mm}{1\\2\\3\\5\\7\\9}}} (6);
 \path (6) edge              node [sloped, text centered, above] {\rlap{~~~$11,12\rightarrow$}} (0);
\end{tikzpicture}
\caption{Graph of transitions between the points  $\pm s_{i}$
under different strong arbitrages. \label{F-1}}
\end{figure}

Ignoring the zero vertex and the edge labels, this graph is
isomorphic to the polyhedral octahedron graph, see Fig.~\ref{sem}.

\begin{figure}[!htbp]
\begin{center}
\includegraphics*[width=0.4\textwidth]{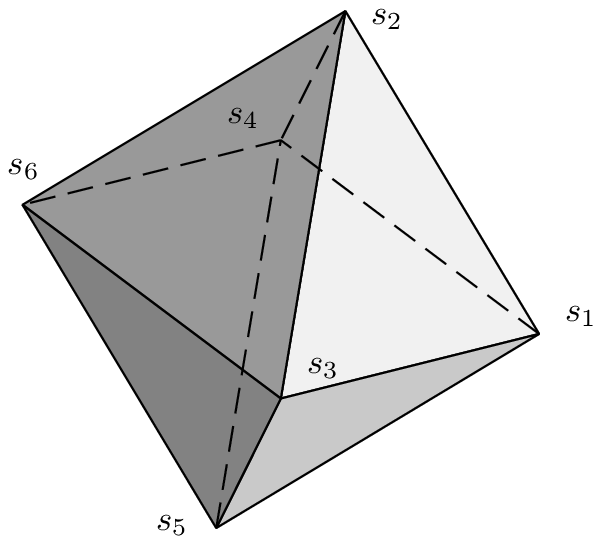}
\hfill
\includegraphics*[width=0.4\textwidth]{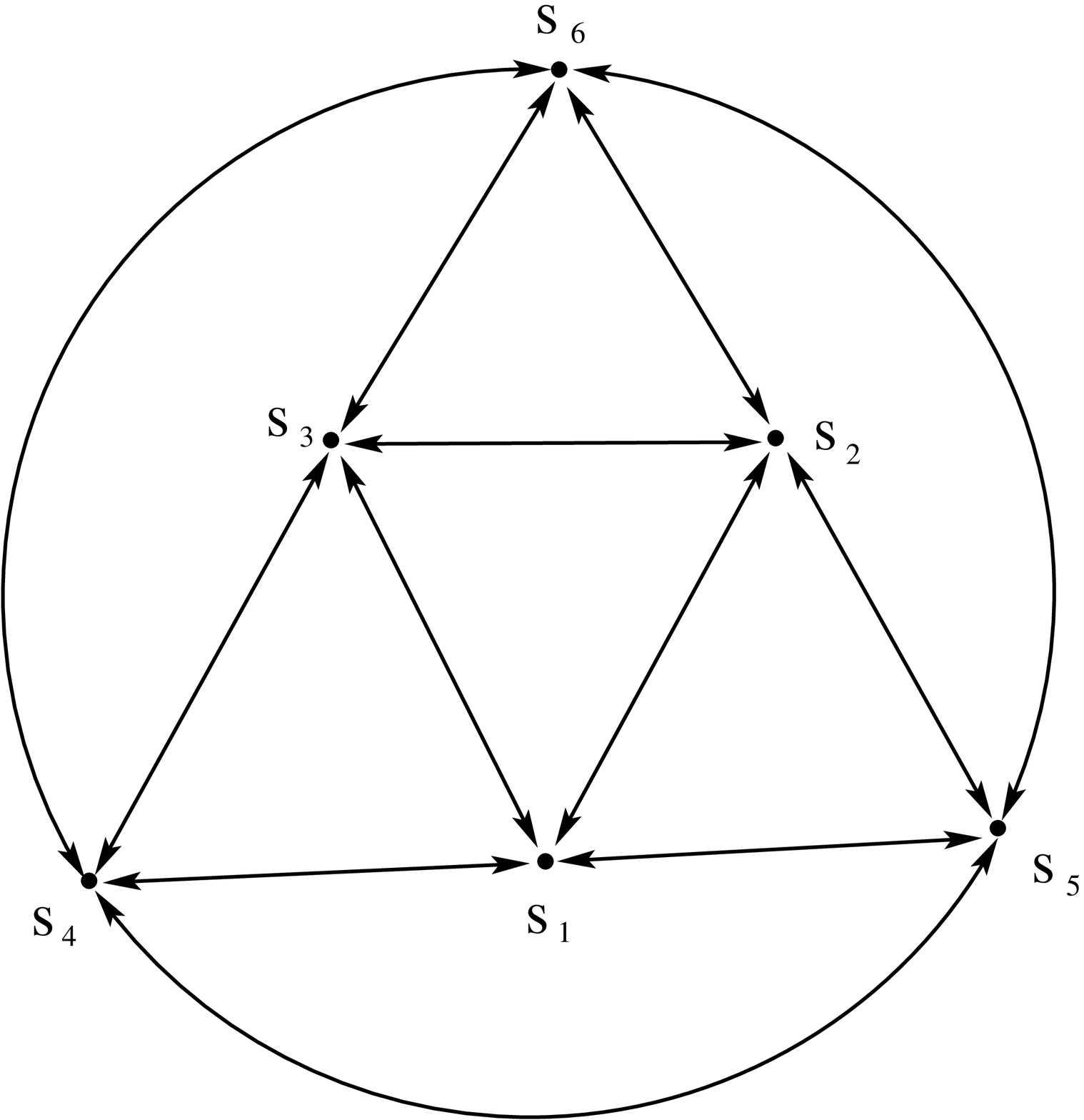}
\hfill~
\caption{The octahedron graph\label{sem}}
\end{center}
\end{figure}

\subsection{Consequences}
By inspection, the set $ \mathbb{P}=\textrm{co}\left\{\pm s_{1},\pm
s_{2},\pm s_{3},\pm s_{4},\pm s_{5},\pm s_{6}\right\} $ has a
non-empty interior. It is a polyhedron, with six quadrilateral and
eight  triangular faces. This polyhedron is a usual (a little bit
elongated) triangular orthobicupola, shown at Fig.~\ref{normF}.
\begin{figure}[!htbp]
\begin{center}
\hfill\includegraphics*{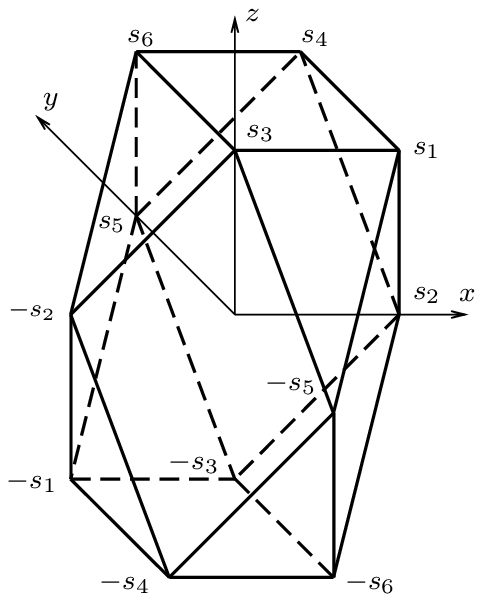}
\hfill
\includegraphics*{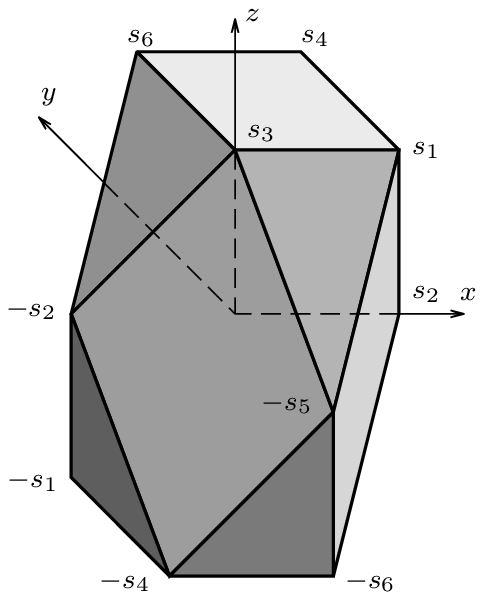}
\hfill~
\caption{The polyhedron $\mathbb{P}$.\label{normF}}
\end{center}
\end{figure}

We can consider this polyhedron $\mathbb{P}$ as the unit ball in an
auxiliary norm $\|\cdot\|_{*}$, in which $ \|G_{n}\|_{*}\le 1,\quad
n=1,2,\ldots, 12. $ Thus, the set of matrices $\{G_{n}\}$ is
neutrally stable.

This implies that any product of matrices $D$, and therefore any
product of matrices $B$  has only eigenvectors which are equal
either to $0$ or to $1$. In particular the spectral radius of any
product is equal to $1$. This proves both Proposition \ref{VictorP}
and Proposition \ref{grop}.


Now we present two interesting types of sequences of
strong arbitrages, that appeared to be useful.

The sequence is $\hat{\bf A}$  is called {\em stabilizer}, if for any $\calR$ the
corresponding outcome $\hat{\bf A}\calR$ is balanced. For example,
 the chain
\[
\hat{\bf A}=A_{AMR}A_{FRA} A_{FMR}
\]
is a stabilizer.

By definition the following assertion holds.
\begin{lemma}
\label{stabL} The chain of arbitrages is a stabilizer if and only if
the product of the corresponding sequence of matrices $G$ is equal
to zero.
\end{lemma}

By $\hat{\bf A}^{p}$ we denote concatenation of $p$ exemplars
of $\hat{\bf A}$. For example, if $\hat{\bf A}=A_{AMR},A_{FRA},$
then
\[
\hat{\bf A}^{3}=A_{AMR} A_{FRA} A_{AMR}
A_{FRA} A_{AMR} A_{FRA}.
\]
The chain $\hat{\bf A}$ is called {\em destabilizer}, if  for some
$\calR$ all elements
\[
\calR\hat{\bf A}^{p}, \quad p=1,2,\ldots ,
\]
are pair-wise different.

\begin{lemma}
\label{destabL} The sequence of arbitrages is a destabilizer if and
only if the product of the corresponding sequence of matrices $D$
is equal an adjoint vector for the eigenvalue 1.
\end{lemma}

Follows from definitions.

The last two lemmas have been used in construction of the sequence
from Proposition \ref{arbP} in the following way. First, we have
chosen a destabilizer given by the following chain of strong
arbitrages:
\[
\hat{\bf A}= \calA_{FRM} \calA_{FMR}  \calA_{RMA}.
\]
Secondly, we multiplied it by the stabilizer
\[
\hat{\bf A}=A_{AMR} A_{FRA}  A_{FMR}.
\]
Thirdly, we have produced the corresponding sequence of arbitrages.
Finally, we found, that in our particular case this ``stabilizing''
part can be reduced to $\calA_{FMA}$.

\subsection{Links to the Asynchronous Systems Theory}

In conclusion, we make three remarks which could be useful in
investigation of systems with more than four producers.

\begin{itemize}\label{Rem3}
\item Construction of matrices $B_{n}$ may be interpreted as a
    special case of construction of {\em  mixtures} of matrices
    in the asynchronous systems theory  (see \cite{AKKK:92:e}
    or a short survey \cite{Koz:UCC}). Convergence analysis for
    the product of matrices $B$ is analogous to analysis of
    absolute $r$-asymptotic stability of asynchronous system.
    This problem is also similar to the problem of estimating
    the generalized (joint) spectral radius of a family of
    matrices consisting of all matrices $B_{n}$.

\item Since the matrices $G$ are integer, the convergence of
    long regular sequences to zero is similar to the well known
    \emph{mortality problem} (see
    \cite{BT:Autom00,BT:IPL97,BT:SCL00,TB:MCSS97}).

\item In the case of matrices $G$, the subspace of common fixed
    points of these matrices is trivial. Moreover, this set of
    matrices is {\em irreducible}: they do not have common
    invariant subspaces. Following
    \cite{Koz:ArXiv08-2,Koz:DAN09:e} one can find an explicit
    estimate for norms of all products of matrices from
    irreducible family. Furthermore, if all entries of the
    matrices are integer, the question whether any sufficiently
    long product would equal to zero is algorithmically
    solvable in a finite (may be very large, but still finite)
    number of operations.
\end{itemize}

\bibliographystyle{ait-en}
\bibliography{balance}
\end{document}